\documentclass[journal]{IEEEtran}
\PassOptionsToPackage{ruled, linesnumbered}{algorithm2e}
\usepackage[latin9]{inputenc}
\usepackage{color}
\usepackage{multirow}
\usepackage{array}
\usepackage{algorithm2e}
\usepackage{amsmath,amsfonts}
\usepackage{amsthm}
\usepackage{amssymb}
\usepackage[caption=false,font=normalsize,labelfont=sf,textfont=sf]{subfig}
\usepackage{graphicx,verbatim}
\usepackage{cite,textcomp}
\usepackage{setspace}
\allowdisplaybreaks 
\makeatletter
\theoremstyle{plain}
\newtheorem{thm}{\protect\theoremname}
\newtheorem{remark}{Remark}

\usepackage{balance}
\usepackage{pgfplots}
\usepackage{pgfplotstable}
\usepackage{tikz}
\usetikzlibrary{calc}
\usetikzlibrary{arrows.meta}
\usetikzlibrary{patterns}

\DeclareMathOperator{\argmax}{argmax}
\usepackage{comment}

\makeatother

\providecommand{\theoremname}{Theorem}

\begin{document}
\title{\LARGE{Multicasting Pinching Antenna Systems With LoS Blockage}}
\author{Muhammad Fainan Hanif, and Yuanwei Liu,~\IEEEmembership{Fellow,~IEEE}
\thanks{M. F. Hanif is with the Institute of Electrical, Electronics and Computer
Engineering, University of the Punjab, Lahore, Pakistan (email: fainan.hanif@gmail.com).}	
\thanks{Y. Liu is with the Department of Electrical and Electronic Engineering, the
University of Hong Kong, Hong Kong (email: yuanwei@hku.hk).}
}
\maketitle
\begin{abstract}
Pinching-antenna systems (PASS) represent a promising customizable wireless access mechanism in high-frequency bands, enabled by dielectric waveguides and movable dielectric particles, called pinching antennas (PAs). In this work, we study optimal position allocation of PAs in PASS for multicasting in the downlink when a line-of-sight (LoS) link does not necessarily exist between all users and the PAs. The multicasting problem is solved by leveraging minorization-maximization (MM) principle to yield a provably convergent algorithm. In each run of the MM based procedure, we solve a convex surrogate problem using two methods called the candidate search method (CSM) and the bisection search method (BSM). With both BSM and CSM, we not only report superior performance of the multicasting PASS in non-LoS environments compared to conventional antenna systems (CAS), but also  determine that BSM yields better overall computational complexity when the number of users and PAs increases. For example, we report that when we have 8 PAs and 25 users, the execution time with the CSM is approximately 2.5 times that with the BSM.
\end{abstract}


\section{Introduction}
To prevail over free space path losses, either the end user can be directly connected with a wire or drawn close to the base station (BS). Naturally, both solutions are impractical to implement in conventional wireless systems. Motivated by this, leaky waveguide based pinching-antenna systems (PASS) were first prototyped and discussed by NTT DOCOMO in the pioneering work by Fukuda \emph{at al.} \cite{Fukuda_DOCOMO}. The PASS are being envisaged to operate in high frequency sixth-generation (6G) networks to enhance communication zones and create them where they did not exist previously. Since the pioneering in depth study by Ding \emph{et al.} \cite{Ding_fundamental}, PASS have gained huge research interest. \par Several aspects of PASS have already been explored in research literature. For instance, in Xu \emph{et al.} \cite{Xu_dwnlink_PASS}, the authors solve the data rate maximization of  downlink pinching antennas (PAs) by proposing a relaxed problem and a two-stage optimization algorithm to solve it. The impact of line-of-sight (LoS) blockage is explored in \cite{Ding_PASS_block} by Ding \emph{et al.}, where it is concluded that LoS blockage is particularly useful in a multiuser scenario using the performance metric of outage probability. Likewise, Wang \emph{at al.} \cite{Wang_LoS_blck} look into the PA activation problem to maximize sum throughput when LoS blockages are considered. PASS aided multicasting of the same message to multiple users is investigated in \cite{Chen_PA_multicast} by Chen \emph{et al.} A cross-entropy framework is proposed in \cite{Chen_PA_multicast} without considering communication barriers in LoS between PAs and users. In the preprint by Shan \emph{et al.} \cite{Shan_multicast}, the authors examine multicasting with preconfigured PA positions with no blockages.\par In this paper, we study the impact of blockages on PA assisted physical layer multicasting system. By modeling the presence or absence of LoS using a Bernoulli random variable, we determine the positions of PAs along a waveguide such that the downlink mean signal-to-noise ratio (SNR) is maximized. Specifically, our contributions include: i) formulation of the main non-convex multicasting problem based on average downlink SNR; ii) rigorous development of a convex surrogate formulation for the multicasting problem to be used together with minorization-maximization (MM) framework; iii) two solutions of the proposed convex problem based on probing possible optimal decision variables and bisection search within each run of the MM approach; iv) comparison of the two methods from different perspectives, and v) numerical evaluation of the performance of the proposed solutions.\newline 
\emph{Notation:} Unless otherwise mentioned, $\|\cdot\|$ represents the Euclidean norm, $c^*$ is the conjugate of the complex number $c$, $\mathbb{E}_X[\cdot]$ denotes the expectation with respect to the random variable $X$, and $\text{Pr}\{E\}$ is the probability of the event $E$. The notation $\mathcal{O}(\cdot)$ represents the worst-case scenario of the algorithm's runtime, and $\nabla$ denotes the gradient operator. Lastly, both $v^{\star}$ and $v^\circ$ are used to designate the optimal values of a decision variable $v$.
\section{System Model and Problem Formulation}\label{SystemModel}
\begin{figure}
\centering
\includegraphics[width=1\columnwidth]{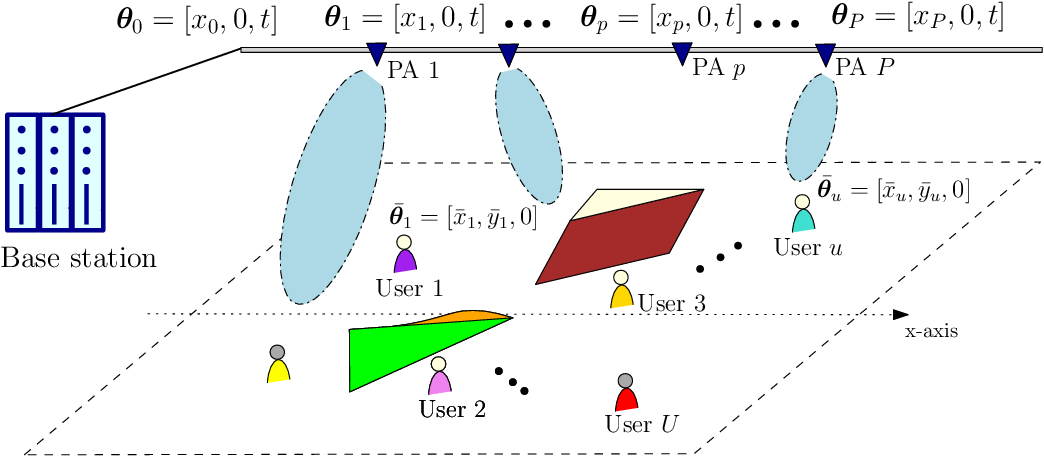}
\caption{Pinching antenna multicast transmission in the presence of blocking objects.} \label{fig:systemmodel}
\end{figure}
We consider a downlink system consisting of a waveguide equipped with $P$ PAs that communicate with $U$ single antenna users as shown in Fig. \ref{fig:systemmodel}. The waveguide is parallel to the horizontal axis located at a height $t$ above the service area such that the common scalar signal with unit power propagates through it. The position of the $p^{\text{th}}$ PA is $\boldsymbol{\boldsymbol{\theta}}_p=[x_p,0,t]$. It is assumed that the user $u$ is located at $\bar{\boldsymbol{\theta}}_u=[\bar{x}_u,\bar{y}_u,0]$. The coordinates of the waveguide feed point are given as $\boldsymbol{\theta}_0=[x_0,0,t]$. The phase shift in the signal from the feed point to the $p^{\text{th}}$ PA is $\phi_p=\frac{2\pi}{\lambda_g}\|\boldsymbol{\theta}_0-\boldsymbol{\theta}_p\|$, where $\lambda_g=\lambda/n_{\text{eff}}$ is the wavelength in the waveguide with an effective refractive
index of $n_{\text{eff}}$, and $\lambda$  is the free-space wavelength.\par In addition to the above, the LoS path between PA $p$ and a user $u$ is modeled by an indicator random variable $\mathcal{I}_{u,p}\in\{0,1\}$, where $1$ and $0$ indicate the presence and absence of LoS between $p$ and $u$, respectively. Hence, the channel between PA $p$ and user $u$ is given by
\begin{align}
\bar{h}_{u,p}=\mathcal{I}_{u,p}h_{u,p}\label{eq:overall_channel_u_p},
\end{align}
where 
\begin{align}
h_{u,p}=\frac{\eta^{1/2}}{\|\bar{\boldsymbol{\theta}}_u-\boldsymbol{\theta}_{p}\|}e^{-j\frac{2\pi}{\lambda}\|\bar{\boldsymbol{\theta}}_u-\boldsymbol{\theta}_{p}\|}e^{-j\frac{2\pi}{\lambda_g}\|\bar{\boldsymbol{\theta}}_0-\boldsymbol{\theta}_{p}\|}\label{eq:channel_u_p},
\end{align}
and where $\eta=\frac{c^2}{16\pi^2 f_c^2}$ represents the free space path loss coefficient where $c$ is the speed of light, and $f_c$ is the carrier frequency. Now, from \cite{Ding_PASS_block,Ding_EDMA}
\begin{align}
\text{Pr}\{\mathcal{I}_{u,p}=1\}=e^{-\alpha\|\bar{\boldsymbol{\theta}}_u-\boldsymbol{\theta}_p\|^2},\label{eq:prob_block}
\end{align}
where $\alpha$ is the LoS blockage parameter. The mean channel strength is therefore obtained as
\begin{align}
\mathbb{E}_{\mathcal{I}_{u,p}}[|\bar{h}_{u,p}|^2]&=\mathbb{E}_{\mathcal{I}_{u,p}}[|{h}_{u,p}|^2|\mathcal{I}_{u,p}=1]\text{Pr}\{\mathcal{I}_{u,p}=1\}\label{eq:avg_channel_power_1}\\& = |{h}_{u,p}|^2e^{-\alpha\|\bar{\boldsymbol{\theta}}_u-\boldsymbol{\theta}_p\|^2}.\label{eq:avg_channel_power_2}
\end{align}
We note that the mean received signal power at the $u^{\text{th}}$ user is given by
\begin{align}
\mathbb{E}[|\sum_{p=1}^P \mathcal{I}_{u,p}h_{u,p}|^2]&=\sum_{p=1}^P \mathbb{E}[|\mathcal{I}_{u,p}h_{u,p}|^2]+\nonumber\\&\sum_{p=1}^P\sum_{p^\prime\neq p}\mathbb{E}[\mathcal{I}_{u,p}\mathcal{I}_{u,p\prime}h_{u,p}h_{u,p^\prime}^*],\label{eq:desired_signal_power}
\end{align}
where the above is the joint expectation over the spatial phase of the channel and the LoS indicator random variable. The first expectation on the right side of \eqref{eq:desired_signal_power} is the same as \eqref{eq:avg_channel_power_2}, and noting the independence of LoS indicator random variable from the spatial phase of channels, observe that:
\begin{align}
\mathbb{E}[h_{u,p}h_{u,p^\prime}^*]=\frac{\eta}{\|\bar{\boldsymbol{\theta}}_u-\boldsymbol{\theta}_{p}\|\|\bar{\boldsymbol{\theta}}_u-\boldsymbol{\theta}_{p^\prime}\|}\mathbb{E}[e^{-j\Delta\Phi(\bar{\boldsymbol{\theta}}_u)}],\label{eq:mean_2_tot_signal_power}
\end{align}
where $\Delta\Phi(\bar{\boldsymbol{\theta}}_u)\triangleq\frac{2\pi}{\lambda}(\|\bar{\boldsymbol{\theta}}_u-\boldsymbol{\theta}_{p}\|-\|\bar{\boldsymbol{\theta}}_u-\boldsymbol{\theta}_{p^\prime}\|)+\frac{2\pi}{\lambda_g}(\|\bar{\boldsymbol{\theta}}_0-\boldsymbol{\theta}_{p}\|-\|\bar{\boldsymbol{\theta}}_0-\boldsymbol{\theta}_{p^\prime}\|)$. Even a slight change in the position of the user $u$ results in rapid fluctuations in the phase term $\Delta\Phi(\bar{\boldsymbol{\theta}}_u)$ which makes it behave as a uniform random variable in $[0,2\pi]$, resulting in zero mean value in \eqref{eq:mean_2_tot_signal_power}. Hence, with the total transmit power $P_{\text{TX}}$ uniformly allocated to all $P$ PAs \cite{Chen_PA_multicast}, the average SNR at $u$ is given by
\begin{align}
\text{SNR}_u=\frac{\rho\sum_{p=1}^P|{h}_{u,p}|^2e^{-\alpha\|\bar{\boldsymbol{\theta}}_u-\boldsymbol{\theta}_p\|^2}}{\sigma^2},\label{eq:avd_SNR}
\end{align}
where $\rho=\frac{\eta P_{\text{TX}}}{P}$ and we have assumed that the additive white Gaussian noise variance $\sigma_u^2=\sigma^2$ for all $u$. The main problem considered in the paper is stated as follows
\begin{subequations}\label{MProb}
\begin{align}
\mathop{{\rm maximize}}\limits _{\mathbf{x}} & \quad\min_{u=1,\ldots,U} \text{SNR}_u\\
{\rm subject~to} & \quad x_p\in [D_1,D_2],\: \forall p\label{eq:x_p_range}\\
 &\quad |x_p-x_q|\geq \delta,\:\forall p\neq q\label{eq:x_p_dist}
\end{align}
\end{subequations}
where $\mathbf{x}\triangleq[x_1,\ldots,x_P]$, the vector with positions of the $P$ PAs along the waveguide as its coordinates, is the decision variable, $D_1$ and $D_2$ signify the endpoints of the waveguide and $\delta$ is the minimum separation between any two PAs. We see that
\begin{align}
\text{SNR}_u=\rho^\prime\sum_{p=1}^P\frac{ e^{-\alpha[(x_p-\bar{x}_u)^2+\bar{y}_u^2+t^2]}}{(x_p-\bar{x}_u)^2+\bar{y}_u^2+t^2},\label{eq:SINR_def}
\end{align}
where $\rho^\prime=\rho/\sigma^2$. Let us define
\begin{align}
f_u(z)=\frac{e^{-\alpha[(z-\bar{x}_u)^2+c_u]}}{(z-\bar{x}_u)^2+c_u},\label{eq:def_f}
\end{align}
where $c_u=\bar{y}_u^2+t^2$. Thus, we have $\text{SNR}_u=\rho^\prime\sum_{p=1}^Pf_u(x_p)\triangleq S_u(\mathbf{x})$. Our original problem therefore becomes
\begin{align}
\mathop{{\rm maximize}}\limits _{\mathbf{x}\in \mathcal{F}} & \quad\min_{u=1,\ldots,U} S_u(\mathbf{x}) \label{Mprob_eq}
\end{align}
where the feasibility set $\mathcal{F}\triangleq \{x_p: x_p\in [D_1,D_2],\forall p, |x_p-x_q|\geq \delta,\forall p\neq q\}$. 
\section{Proposed Solution}
The main problem in \eqref{Mprob_eq} is non-convex with no straightforward solution. In order to solve \eqref{Mprob_eq}, we first develop some necessary mathematical background. Towards this end, consider the following property of $f_u(z)$.
\begin{thm}
\label{thm1:1stApp} The function $f_u(t)$ satisfies
\begin{align}
f_u(z)\geq f_u(z_0)+B_u(z_0)[(z-\bar{x}_u)^2-(z_0-\bar{x}_u)^2]\label{ineq:transformed_supp_hyperplane}
\end{align}
where $z_0$ is the given point and
\begin{align}
B_u(z_0)\triangleq \frac{-e^{-\alpha\{(z_0-\bar{x}_u)^2+c_u\}}[\alpha\{(z_0-\bar{x}_u)^2+c_u\}+1]}{\{(z_0-\bar{x}_u)^2+c_u\}^2}.\label{eq:const_B_u(t)}
\end{align}
\end{thm}
\begin{IEEEproof}
Note that $q\triangleq(z-\bar{x}_u)^2+c_u$ is positive when $t>0$, and let $f_u(z)\triangleq \varphi(q)=\frac{e^{-\alpha q}}{q}$. The convexity of $\varphi(q)$ follows by noting that the second derivative of $\ln \varphi(q)$ is $q^{-2}$,which is non-negative. Hence, the supporting hyperplane inequality of convex functions implies 
\begin{align}
\varphi(q) \geq \varphi(q_0)+\varphi^\prime(q_0)(q-q_0)\label{ineq:supp_hyperplane}
\end{align}
where $q_0$ is a given point and $\varphi^\prime(q)=\frac{-e^{-\alpha q}(\alpha q + 1)}{q^2}$. It is observed that for all $q>0$ and $\alpha\geq 0$, the first derivative $\varphi^\prime(q)$ is negative. Using $q_0=(z_0-\bar{x}_u)^2+c_u$, the inequality in \eqref{ineq:supp_hyperplane} is transformed to the desired form in \eqref{ineq:transformed_supp_hyperplane}.
\end{IEEEproof}
\begin{remark}\label{remark:MM_framework_requirements}
Here we note that when $z=z_0$, the inequality in \eqref{ineq:transformed_supp_hyperplane} is satisfied with equality. Moreover, it can be shown straightforwardly that the derivative of both sides of \eqref{ineq:transformed_supp_hyperplane} is also the same at $z=z_0$. These two properties play a critical role in the convergence of iterative algorithms based on the MM framework \cite{Hunter_MM_tutorial}.
\end{remark}
Indeed, it is interesting to observe that, regardless of the properties of $f_u(z)$, the inequality in \eqref{ineq:transformed_supp_hyperplane} is a consequence of the convexity of $\varphi(q)$. Now from \eqref{ineq:transformed_supp_hyperplane}, it is deduced that
\begin{align}
&\rho^\prime\sum_{p=1}^P f_u(x_p)\geq  \rho^\prime\sum_{p=1}^P f_u(x_p^{(n)})\nonumber\\&+\rho^\prime\sum_{p=1}^P B_u(x_p^{(n)})[(x_p-\bar{x}_u)^2-(x_p^{(n)}-\bar{x}_u)^2]\label{ineq:sum_f}\\& \Leftrightarrow S_u(\mathbf{x}) \geq S_u(\mathbf{x}^{(n)})\nonumber\\&+\rho^\prime\sum_{p=1}^P B_u(x_p^{(n)})[(x_p-\bar{x}_u)^2-(x_p^{(n)}-\bar{x}_u)^2]\triangleq L_u(\mathbf{x}),\label{ineq:S_u}
\end{align}
where we have invoked \eqref{ineq:transformed_supp_hyperplane} at $z_0=x_p^{(n)}$, and the superscript $(n)$ corresponds to the iteration number of the algorithm we propose next. Furthermore, also note that $L_u(\mathbf{x})=\bar{L}_u(\mathbf{x}^{(n)})+\rho^\prime\sum_{p=1}^P B_u(x_p^{(n)})[(x_p-\bar{x}_u)^2]$, where $\bar{L}_u(\mathbf{x}^{(n)})\triangleq S_u(\mathbf{x}^{(n)})-\rho^\prime\sum_{p=1}^P B_u(x_p^{(n)})[(x_p^{(n)}-\bar{x}_u)^2]$.
Hence, we replace \eqref{Mprob_eq} with the following surrogate problem
\begin{align}
\mathop{{\rm maximize}}\limits _{\mathbf{x}\in \mathcal{F}} & \quad\min_{u=1,\ldots,U} L_u(\mathbf{x}). \label{Mprob_eq_lb}
\end{align}
Specifically, in the $(n+1)^{\text{st}}$ run of the MM procedure, developed as \textbf{Algorithm \ref{alg:multicast_PA_blck}}, we solve \eqref{Mprob_eq_lb}. To do so, we exploit the properties of the upper bound $L_u(\mathbf{x})$ and develop two solutions of \eqref{Mprob_eq_lb} in the next sections.
\subsection{Candidate Search Method (CSM) for Maximization of Lower Envelope of Inverted Parabolas}\label{sec:parabolas}
From the structure of $L_u(\mathbf{x})$, we see that it is separable in $x_p$. Therefore, we perform a coordinate wise maximization of the objective in \eqref{Mprob_eq_lb}. To find the position of $p^{\text{th}}$ PA, define $\tilde{L}_u(\mathbf{x}^{(n)})\triangleq \bar{L}_u(\mathbf{x}^{(n)})+\rho^\prime\sum_{q\neq p} B_u(x_q^{(n)})[(x_q-\bar{x}_u)^2$ for all $u$. Hence, $L_u(\mathbf{x})=\tilde{L}_u(\mathbf{x}^{(n)})+\rho^\prime B_u(x_p^{(n)})(x_p-\bar{x}_u)^2$. As $B_u(x_p^{(n)})$ is negative, we see that \eqref{Mprob_eq_lb} represents the maximization of the lower envelope of inverted parabolas over $\mathcal{F}$. In order to maximize the objective in \eqref{Mprob_eq_lb} for each $x_p$, we determine the feasible region along the length of the waveguide inside which the $p^{\text{th}}$ PA can be placed such that no two PAs are within a distance of $\delta$ from each other. Since we perform coordinate wise maximization, the locations of all $x_q$ where $q\neq p$ are known, and the constraints $|x_p-x_q|\geq \delta$ need to be satisfied for $p\neq q$. Therefore, the feasible intervals for $x_p$ are obtained by subtracting from $[D_1,D_2]$ the union of all intervals of the form $(x_q-\delta,x_q+\delta)$ for all $q\neq p$. As a result, we ensure no $x_q$ is within a distance $\delta$ from the $x_p$ to be optimized. The feasible set for $x_p$ thus becomes $\mathcal{F}_p\in [D_1,D_2]\setminus \bigcup_{q\neq p} (x_q-\delta,x_q+\delta)$, which can be a disjoint union of intervals.\par Once the feasible intervals are known, for each interval, the set of candidate points, $\mathcal{C}_m$, where the minimum of $L_u({\mathbf{x}})$ over $u$ is maximized, is determined. To determine $\mathcal{C}_m$ for downward pointing and intersecting parabolas, led by intuition and the extreme value theorem, three possibilities are investigated: (i) end points of the interval, (ii) vertices such that $x_p=\bar{x}_u$ of all $U$ parabolas subject to they belong to the interval, and (iii) intersection points where parabolas pairwise intersect. Point (iii) implies that the candidate optimal points satisfy $L_u(\mathbf{x})=L_w(\mathbf{x})$ for all $w\neq u$, each of which is a quadratic equation whose real solution which belongs to the given interval is taken into account. For each interval that constitutes $\mathcal{F}_p$, we take those points from $\mathcal{C}_m$ where the minimum of $L_u(\mathbf{x})$ over all $U$ users is attained. More specifically, let the $R$ potential locations of $x_p$ be denoted by $x_{\text{cand}}^1,\ldots,x_{\text{cand}}^R$. For each candidate $x_{\text{cand}}^r$, we evaluate $h_u(x_{\text{cand}}^r)\triangleq \tilde{L}_u(\mathbf{x}^{(n)})+\rho^\prime B_u(x_p^{(n)})(x_{\text{cand}}^r-\bar{x}_u)^2$ for all $u$. The minimum of $h_u(x_{\text{cand}}^r)$ is stored in $H(x_{\text{cand}}^m)$. The same procedure is carried out for all intervals in $\mathcal{F}_p$. Finally, in the $(n+1)^{\text{st}}$ iteration, $x_p^{(n+1)}$ is chosen as that candidate solution which maximizes $H(x_{\text{cand}}^m)$ over all intervals. This method is outlined as \texttt{Procedure A} in \textbf{Algorithm \ref{alg:proced_A}}.
\subsection{Bisection Search Method (BSM) for Coordinate Ascent}
It is seen from \eqref{ineq:S_u} that $S_u(\mathbf{x})$ is lower bounded by $\bar{L}_u(\mathbf{x}^{(n)})+\rho^\prime\sum_{p=1}^P B_u(x_p^{(n)})[(x_p-\bar{x}_u)^2]$. Recall that within the MM framework, we solve \eqref{Mprob_eq_lb} using coordinate wise maximization. To obtain $x_p^{(n+1)}$, we consider the following convex optimization problem:
\begin{align}
\mathop{{\rm maximize}}\limits _{x_p\in \mathcal{F},s} & \: s \text{ s. t. } \tilde{L}_u(\mathbf{x}^{(n)})+\rho^\prime B_u(x_p^{(n)})(x_p-\bar{x}_u)^2\geq s,\forall u\label{Mprob_eq_coordinate_wise}
\end{align}
where $s$ is an analysis variable used to obtain the epigraph form in \eqref{Mprob_eq_coordinate_wise} and $\tilde{L}_u(\mathbf{x}^{(n)})=\bar{L}_u(\mathbf{x}^{(n)})+\rho^\prime\sum_{q\neq p} B_u(x_q^{(n)})[(x_q-\bar{x}_u)^2$. The constraint corresponding to user $u$, i.e.,
\begin{align}
\tilde{L}_u(\mathbf{x}^{(n)})+\rho^\prime B_u(x_p^{(n)})(x_p-\bar{x}_u)^2\geq s\label{ineq:constraint_u}
\end{align}
implies that 
\begin{align}
|x_p-\bar{x}_u| \leq \sqrt{\frac{\tilde{L}_u(\mathbf{x}^{(n)})-s}{-\rho^{\prime} B_u(x_p^{(n)})}}\triangleq Q_u(s). \label{ineq:constraint_u_eq1}
\end{align}
In turn, the above implies that the feasible set for each user $u$, corresponding to $x_p$ is given by
\begin{align}
\mathcal{F}_p^u\triangleq [\bar{x}_u-Q_u(s),\bar{x}_u+Q_u(s)].\label{eq:feasible_set_u}
\end{align}
It follows from \eqref{eq:feasible_set_u} that for all $u$, the validity of \eqref{ineq:constraint_u} for all users requires $x_p\in \bigcap_{u=1}^U \mathcal{F}_p^u$. This is the same as mandating that $x_p\geq \bar{x}_u-Q_u(s),\forall u$ or $x_p\geq \max_u(\bar{x}_u-Q_u(s))\triangleq \mathtt{l}^\prime(s)$. Similarly, $x_p\leq \min_u(\bar{x}_u+Q_u(s))\triangleq \mathtt{u}^\prime(s)$. Hence, $x_p\in [\mathtt{l}^\prime(s),\mathtt{u}^\prime(s)]$, which is nonempty when $\mathtt{l}^\prime(s)\leq \mathtt{u}^\prime(s)$. The feasible interval $\mathcal{F}_p$ for $x_p$, as discussed and determined in Sec. \ref{sec:parabolas}, can be written as $\mathcal{F}_p=\bigcup_{m=1}^M[\mathtt{a}_m,\mathtt{b}_m]$. Hence, overall the position of the $p^{\text{th}}$ PA satisfies $x_p\in \mathcal{F}_p\cap [\mathtt{l}^\prime(s),\mathtt{u}^\prime(s)]$, which translates to $x_p\in[\max(\mathtt{l}^\prime(s),\mathtt{a}_m),\min(\mathtt{u}^\prime(s),\mathtt{b}_m)]$.\par In order to determine $s$ in \eqref{Mprob_eq_coordinate_wise} using bisection search, the maximum $s_{\text{max}}$ and minimum $s_{\text{min}}$ values of $s$ need to be known. From \eqref{ineq:constraint_u}, due to negative $B_u(x_p^{(n)})$, it follows that $s\leq \tilde{L}_u(\mathbf{x}^{(n)}),\forall u$ which results in $s_{\text{max}}=\min_u \tilde{L}_u(\mathbf{x}^{(n)})$. The parameter $s_{\text{min}}$ can be determined as $s_{\text{min}}=\max_{x_p\in\mathcal{C}_{\text{cand}}}\min_u \{\tilde{L}_u(\mathbf{x}^{(n)})+\rho^\prime B_u(x_p^{(n)})(x_p-\bar{x}_u)^2\}$, where $\mathcal{C}_{\text{cand}}$ represents the set of candidate points determined in Sec. \ref{sec:parabolas}. It is important to emphasize that in $s_{\text{min}}$, we do not need to check all three categories of candidate points exhaustively, and finding $s_{\text{min}}$ for a subset of $\mathcal{C}_m$ is sufficient. We summarize this bisection search based per coordinate maximization approach in \texttt{Procedure B} as \textbf{Algorithm \ref{alg:proced_B}}. The optimal position $x_p^{\circ}$ is determined as the midpoint of the interval $[\max(\mathtt{l}^\prime(s),\mathtt{a}_m),\min(\mathtt{u}^\prime(s),\mathtt{b}_m)]$ at $s_{\text{min}}$, which, by virtue of the construction of \textbf{Algorithm \ref{alg:proced_B}}, is the highest value of $s$ with a proven existence of feasible $x_p$.
\begin{algorithm}
\caption{Pinching antennas enabled blockage aware multicasting algorithm}
\label{alg:multicast_PA_blck}
\KwIn{$D_1,D_2,\delta,c,f_c,\bar{\boldsymbol{\theta}}_0,\bar{\boldsymbol{\theta}}_u, \forall u, t, (\bar{x}_u,\bar{y}_u),\forall u$}
\KwOut{ $x_p^\star,\forall p$}
\For{N times}{
initialize with $x_p^{(0)},\forall p$\;
$n\leftarrow0$\;
\Repeat{convergence }{
\texttt{Procedure A} or \texttt{Procedure B} to get $x_p^{\circ}$\;
Update: $x_p^{(n+1)}\leftarrow x_{p}^{\circ},\forall p$\;
$n\leftarrow n+1$}}
\end{algorithm}
\begin{algorithm}
\caption{\texttt{Procedure A}}
\label{alg:proced_A}
\KwIn{$\mathcal{F}_p$}
\KwOut{optimal $x_p^{\circ}$ for the given data}
For each interval $I_m\in\mathcal{F}_p$ find $\mathcal{C}_m$\;
For each $x_{\text{cand}}^r\in \mathcal{C}_m$ find $H(x_{\text{cand}}^m)=\min_u h_u(x_{\text{cand}}^r)$\;
$x_p^m=\argmax H(x_{\text{cand}}^m)$ and $\bar{H}_m=H(x_p^m)$\;
$m^\star=\argmax \bar{H}_m$ set $x_p^{\circ}=x_p^{m^\star}$
\end{algorithm}
\begin{algorithm}
\caption{\texttt{Procedure B}}
\label{alg:proced_B}
\KwIn{$\mathcal{F}_p,s_{\text{min}},s_{\text{max}}$}
\KwOut{optimal $s^{\circ}$ and $x_p^{\circ}$ of \eqref{Mprob_eq_coordinate_wise} }
\While{$s_{\max}-s_{\min}>\tau$}{
$s_{\text{mid}}=\frac{s_{\text{max}}+s_{\text{min}}}{2}$\;
\eIf{$[\mathtt{a}_m,\mathtt{b}_m]\cap [\mathtt{l}^\prime(s_{\mathrm{mid}}),\mathtt{u}^\prime(s_{\mathrm{mid}})]\neq\phi$}
{set $s_{\text{min}}=s_{\text{mid}}$\;}{set $s_{\text{max}}=s_{\text{mid}}$\;}
}
$s^{\circ}=s_{\text{mid}}$\;$x_p^{\circ}$: midpoint of $[\max(\mathtt{l}^\prime(s_{\text{min}}),\mathtt{a}_m),\min(\mathtt{u}^\prime(s_{\text{min}}),\mathtt{b}_m)]$
\end{algorithm}
\subsection{MM Framework and Complexity Comparison}\label{MM-complexity}
Once the optimal $\mathbf{x}$ in \eqref{Mprob_eq_lb} has been determined using either \texttt{Procedure A} or \texttt{Procedure B}, we use it to update the positions of all PAs in the $(n+1)^{\text{st}}$ iteration of the MM iterative procedure, as shown in \textbf{Algorithm \ref{alg:multicast_PA_blck}}. As noted in \textbf{Remark \ref{remark:MM_framework_requirements}}, the surrogate objective in \eqref{Mprob_eq_lb} satisfies $L_u(\mathbf{x}^{(n)})=S_u(\mathbf{x}^{(n)})$ and equality of $\nabla L_u(\mathbf{x})$ with $\nabla S_u(\mathbf{x})$ at $\mathbf{x}^{(n)}$, which are sufficient to establish the convergence of the MM based procedure to the Karush-Kuhn-Tucker (KKT) conditions of \eqref{Mprob_eq}. Due to similarity in arguments as given in \cite{Hunter_MM_tutorial,Marks_inner_apprx}, and space reasons, the details have not been included. \par We provide a comparison of the worst-case per-iteration complexity of both methods within the MM loop. In \texttt{Procedure A}, for each interval, we have $\mathcal{O}(U^2)$ pairwise comparisons and $\mathcal{O}(U)$ evaluations, resulting in $\mathcal{O}(U^3)$ work per interval. With $M$ intervals and $P$ positions to be updated, the per-iteration complexity of the MM algorithm thus becomes $\mathcal{O}(MPU^3)$. In \texttt{Procedure B}, for each interval there are $\mathcal{O}(\log_2(\frac{s_{\text{max}}-s_{\text{min}}}{2}))$ bisection iterations each of which involves $\mathcal{O}(U)$ feasibility checks. Hence, for the given $P$ PAs and $M$ intervals, the per-iteration worst-case complexity of the MM method becomes $\mathcal{O}(MPU\log_2(\frac{s_{\text{max}}-s_{\text{min}}}{2}))$ when \textbf{Algorithm \ref{alg:proced_B}} is used. Therefore, clearly, for large number of users $U$, \texttt{Procedure B} is the preferred choice for implementation. Since the MM procedure is sensitive to the starting point, we run the iterative algorithm $N$ times with different initial points to determine the best objective as shown in \textbf{Algorithm \ref{alg:multicast_PA_blck}}.
\section{Numerical Results}
\begin{figure}[tb]
		\centering
\includegraphics[width=1\columnwidth]{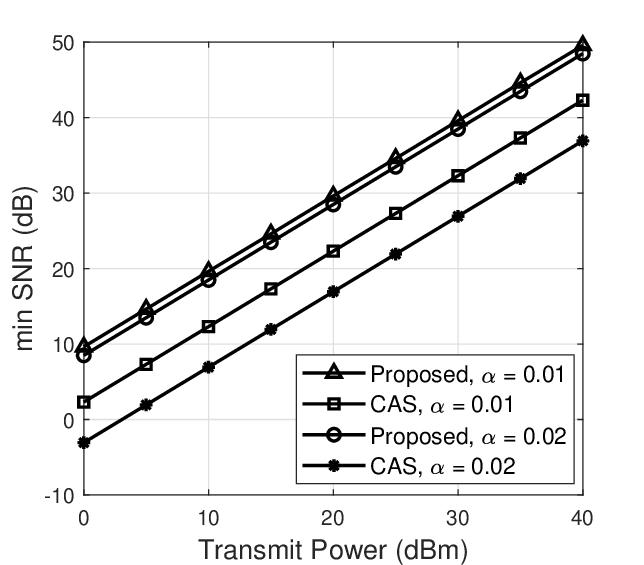}
		\caption{Variation of the objective min. SNR with transmit power for different $\alpha$. We fix $U=5$ and $P=5$.}
		\label{fig: PowVsminSNR}
	\end{figure}
    \begin{figure}[tb]
		\centering
\includegraphics[width=1\columnwidth]{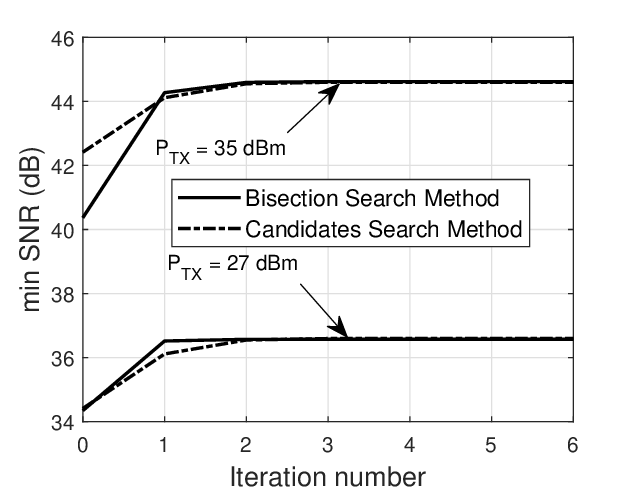}
		\caption{Convergence behaviour of \textbf{Algorithm \ref{alg:multicast_PA_blck}} for different $P_{\text{TX}}$. The parameters $\alpha=0.01,P=5$ and $U=5$ are fixed.}
		\label{fig: convergence}
	\end{figure}
    \begin{figure}[tb]
		\centering
\includegraphics[width=1\columnwidth]{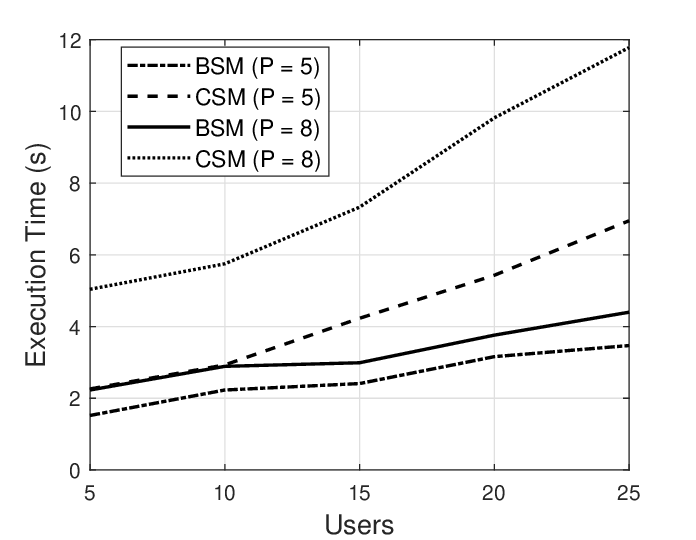}
		\caption{Variation of execution time of \textbf{Algorithm \ref{alg:multicast_PA_blck}} for different $U$. The parameters $U=5,P_{\text{TX}}=40$ dBm and $\alpha=0.01$ are fixed.}
		\label{fig: UsersVsExecTime}
	\end{figure}
In order to numerically assess the proposed framework, we consider a waveguide consisting of $5$ PAs, placed at a height $t=3$ m above a rectangular service area whose width is $10$ m and length is four times its width, such that $5$ single antenna users are uniformly distributed within this region. The transmission frequency is set as 28 GHz, the noise power is $-90$ dBm, $\delta$ is taken as $\lambda/2$, and unless otherwise mentioned $D_1$ and $D_2$ are taken as $-10$ m and $10$ m, respectively. For conventional antenna systems (CAS), we assume the antennas are fixed at the center of the service area with an equal spacing of $\lambda/2$.\par In Fig. \ref{fig: PowVsminSNR}, we study the comparison of the objective function i.e., minimum SNR for both the proposed and the CAS methods as a function $P_{\text{TX}}$ for a couple of LoS blockage parameter, $\alpha$, values. The gap between the proposed and the CAS values is considerably large for both $\alpha$. However, predictably, the minimum SNR increases with $P_{\text{TX}}$ for both types of transmit antennas. It is also seen that a decreased LoS blockage parameter results in a higher achievable worst-case SNR. \par Fig. \ref{fig: convergence} shows the convergence of the main MM based procedure from both CSM and BSM approaches. Within three to four iterations, both CSM and BSM reach the same value of SNR for the transmit powers shown in Fig. \ref{fig: convergence}. \par In Fig. \ref{fig: UsersVsExecTime}, we further compare both BSM and CSM from the perspective of their contribution in execution times of the overall \textbf{Algorithm \ref{alg:multicast_PA_blck}}, as a function of users $U$ for different values of $P$. The trend of graphs shown in Fig. \ref{fig: UsersVsExecTime} corroborate the complexity estimates provided in Sec. \ref{MM-complexity}. Specifically, BSM performs better than the CSM as $U$ increases. In fact, the difference between the execution times of BSM and CSM is significantly enhanced when both $P$ and $U$ have higher values.
\section{Conclusion}
In this paper, we have investigated the design and performance of a downlink multicasting PASS when users do not necessarily have a LoS link between them and PAs. The main optimization problem is solved using MM framework for which novel bounds on the objective function are established. In each iteration of the MM algorithm, we solve the proposed convex surrogate problem using CSM and BSM techniques. It is established in the numerical results section that our solution is not only superior to the CAS method but also possesses favorable complexity properties for practical implementation purposes.
\label{Conc} 

\balance
%
\bibliographystyle{IEEEtran}
\bibliography{IEEEabrv,PA_multicast_block}
\end{document}